\documentclass{amsart}

\usepackage{latexsym,amsfonts} 
\usepackage{amsmath,amssymb,graphics,setspace}
\usepackage{amsthm}
\usepackage{MnSymbol}
\usepackage{mathrsfs} 
\usepackage{fontenc}%

\usepackage{paralist}
\usepackage{color}

\usepackage[verbose]{wrapfig}

\usepackage[ruled,vlined,linesnumbered]{algorithm2e}

\usepackage{marvosym}

\usepackage{picins,graphicx}

\usepackage{pstricks,pst-text,pst-grad,pst-node,pst-3dplot,pstricks-add,pst-poly,pst-coil} 

\usepackage{subfigure}
\renewcommand{\thesubfigure}{\thefigure.\arabic{subfigure}}
\makeatletter
\renewcommand{\p@subfigure}{}
\renewcommand{\@thesubfigure}{\thesubfigure:\hskip\subfiglabelskip}
\makeatother

\usepackage{hyperref}
\hypersetup{linktocpage=true,colorlinks=true,linkcolor=blue,citecolor=blue,pdfstartview={XYZ 1000 1000 1}}

\usepackage{floatflt,graphicx}

\newcommand{\cl}{\mbox{cl}}
\newcommand{\Int}{\mbox{int}}
\newcommand{\bdy}{\mbox{bdy}}
\newcommand{\fil}{\mbox{fil}}
\newcommand{\Nrv}{\mbox{Nrv}}
\newcommand{\near}{\delta} 
\newcommand{\dnear}{\delta_{\Phi}} 
\newcommand{\dcap}{\mathop{\cap}\limits_{\Phi}} 
\newcommand{\assign}{\mathrel{\mathop :}=}

\newcommand{\sk}{\mbox{sk}}
\newcommand{\maxNrvClu}{\mbox{maxNrvClu}}
\newcommand{\maxNerv}{\mbox{maxNrv}}
\newcommand{\clNrv}{\mbox{clNrv}}
\newcommand{\clNrvt}{\mbox{clNrv}_{\Delta}}
\newcommand{\dclNrvt}{\mbox{clNrv}_{\mathop{\Delta}\limits_{\Phi}}}

\newcommand{\dfar}{{\not\delta}_{\Phi}} 
\newcommand{\sn}{\mathop{\delta}\limits^{\doublewedge}} 

\newcommand{\snd}{\mathop{\delta_{_{\Phi}}}\limits^{\doublewedge}} 
\newcommand{\sdfar}{\stackrel{\not{\text{\normalsize$\delta$}}}{\text{\tiny$\doublevee$}}_{_{\Phi}}} %
\newcommand{\sfar}{\stackrel{\not{\text{\normalsize$\delta$}}}{\text{\tiny$\doublevee$}}} 
\newcommand{\notfar}{\mathop{\not{\delta}}\limits^{\doublewedge}} 

\newtheorem{example}{Example}
\newtheorem{remark}{Remark}
\newtheorem{definition}{Definition}
\newtheorem{lemma}{Lemma}
\newtheorem{theorem}{Theorem}
\newtheorem{proposition}{Proposition}
\newtheorem{corollary}{Corollary}

\usepackage{xcolor}
\definecolor{light}{gray}{0.80}

\begin{document}

\title[Proximal Nerve Complexes]{Proximal Nerve Complexes.\\ A Computational Topology Approach}

\author[J.F. Peters]{J.F. Peters}
\email{James.Peters3@umanitoba.ca}
\address{\llap{$^{\alpha}$\,}Computational Intelligence Laboratory,
University of Manitoba, WPG, MB, R3T 5V6, Canada and
Department of Mathematics, Faculty of Arts and Sciences, Ad\.{i}yaman University, 02040 Ad\.{i}yaman, Turkey}
\thanks{The research has been supported by the Natural Sciences \&
Engineering Research Council of Canada (NSERC) discovery grant 185986 
and Instituto Nazionale di Alta Matematica (INdAM) Francesco Severi, Gruppo Nazionale per le Strutture Algebriche, Geometriche e Loro Applicazioni grant 9 920160 000362, n.prot U 2016/000036.}


\date{}

\dedicatory{Dedicated to K. Borsuk and Som Naimpally}

\begin{abstract}
This article introduces a theory of proximal nerve complexes and nerve spokes, restricted to the triangulation of finite regions in the Euclidean plane.  A nerve complex is a collection of filled triangles with a common vertex, covering a finite region of the plane.  Structures called $k$-spokes, $k\geq 1$, are a natural extension of nerve complexes.  A $k$-spoke is the union of a collection of filled triangles that pairwise either have a common edge or a common vertex.  A consideration of the closeness of nerve complexes leads to a proximal view of simplicial complexes.   A practical application of proximal nerve complexes is given, briefly, in terms of object shape geometry in digital images.
\end{abstract}

\maketitle

\section{Introduction}
This article introduces a proximal computational topology approach in the theory of nerve complexes.  
Computational topology combines geometry, topology and algorithms in the study of topological structures, introduced by H. Edelsbrunner and J.L. Harer~\cite{Edelsbrunner2010compTop}.  K. Borsuk was one of the first to suggest studying sequences of plane shapes in his theory of shapes~\cite{Borsuk1970theoryOfShape}.  Borsuk also observed that every polytope can be decomposed $X$ into a finite sum of elementary simplexes, which he called brics.  A \emph{\bf polytope} is the intersection of finitely many closed half spaces~\cite{Ziegle2007polytopes}. This leads to a simplicial complex $K$ covered by simplexes $\Delta_1,\dots,\Delta_n$ (filled triangles) such that the nerve of the decomposition is the same as $K$~\cite{Borsuk1948FMsimplexes}.  Briefly, a \emph{\bf geometric simplicial complex} (denoted by $\Delta(S)$ or simply by $\Delta$) is the convex hull of a set of points $S$, {\em i.e.}, the smallest convex set containing $S$.  Geometric simplexes in this paper are restricted to vertices (0-simplexes), line segments (1-simplexes) and filled triangles (2-simplexes) in the Euclidean plane, since our main interest is in the extraction of features of simplexes superimposed on planar digital images.  In this paper, we consider only what is known as a Vietoris-Rips complex, which is a collection of 2-simplices determined by subsets of 3 points in a set of points in the Euclidean plane~\cite{ChambersSilvaEricksonGhrist2010DCGVietorisRipsComplexes}. 
An important form of simplicial complex is a collection of simplexes called a nerve.

\setlength{\intextsep}{0pt}
\begin{wrapfigure}[13]{R}{0.50\textwidth}
\begin{minipage}{3.8 cm}
\centering
\begin{pspicture}
(-1.0,-0.5)(5.0,4.2)
\psframe[linecolor=black](-0.5,-0.3)(4.3,4.0)
\providecommand{\PstPolygonNode}{%
 \psdots[dotstyle=o,dotsize=0.08,linecolor=blue,fillcolor=yellow](1;\INode)
 \psline(0.95;\INode)}
\PstPolygon[unit=1.75,PolyNbSides=8,fillstyle=solid,fillcolor=lightgray]
\psdot[dotstyle=o,dotsize=0.15,linecolor=blue,fillcolor=yellow](-1.75,1.75)
\rput(-3.8,3.7){$\boldsymbol{X}$}
\rput(-1.0,3.5){$\boldsymbol{\Nrv K}$}
\rput(-0.8,1.5){$\boldsymbol{\sk A}$}
\end{pspicture}
\caption[]{ $\boldsymbol{\Nrv K}$}
\label{fig:spokes}
\end{minipage}
\end{wrapfigure}  

A planar simplicial complex $K$ is a \emph{\bf nerve}, provided the simplexes in $K$ have nonempty intersection (called the nucleus of the nerve).   A nerve of a simplicial complex $K$ (denoted by $\Nrv K$) in the triangulation of a plane region is defined by
$
\Nrv K = \left\{\Delta\subseteq K: \bigcap \Delta\neq \emptyset\right\}\ \mbox{(Nerve complex)}.
$
In other words, the simplexes in a nerve have proximity to each other, since they share the nucleus.  The \emph{\bf nucleus} of a nerve complex is a vertex common to the 2-simplexes in a nerve.  Triangulation of point clouds in the plane provides a straightforward basis for the study of nerve complexes.  A \emph{\bf spoke} $A$ (denoted by $\sk A$) on a nerve complex is a 2-simplex in the nerve.  

\begin{example}
Let $X$ be a planar triangulated region containing a nerve complex $\Nrv K$.   Each filled triangle in $\Nrv K$ is a spoke. For example, $sk A$ in Fig.~\ref{fig:spokes} is a spoke in $\Nrv K$.
\qquad \textcolor{blue}{\Squaresteel}
\end{example}

The study of nerves was introduced by P. Alexandroff~\cite{Alexandroff1928dimensionsbegriff}, elaborated by K. Borsuk~\cite{Borsuk1948FMsimplexes}, J. Leray~\cite{Leray1946homology}, and a number of others such as M. Adamaszek et al.~\cite{Adamaszek2014arXivNerveComplexes}, E.C. de Verdi\`{e}re et al.~\cite{Colin2012multinerves}, H. Edelsbrunner and J.L. Harer~\cite{Edelsbrunner2010compTop}, and more recently by M. Adamaszek, H. Adams, F. Frick, C. Peterson and C. Previte-Johnson~\cite{Adamaszek2016DCGnerveComplexesOfCircularArcs}.  In this paper, an extension of the Borsuk Nerve Theorem is given.

\begin{theorem}\label{thm:nerveTheorem}{\bf\rm Borsuk Nerve Theorem~\cite{Borsuk1948FMsimplexes}}
If $U$ is a collection of subsets in a topological space, the nerve complex is homotopy equivalent to the union of the subsets.
\end{theorem}

A main result in this paper is the following extension of Theorem~\ref{thm:nerveTheorem}.

\begin{theorem}\label{thm:nerveSpokeTheorem}
If $\Nrv K$ is a nerve complex in a topological space, $\Nrv K$ is homotopy equivalent to the union of its $n$-spokes, $n \geq 1$.
\end{theorem}

A practical application of simplicial complexes is the study of the characteristics of surface shapes.  Such shapes can be quite complex when they are found in digital images.  By covering a part or all of a digital image with simplexes, we simplify the problem of describing object shapes, thanks to a knowledge of geometric features of either individual simplices or simplicial complexes.  The problem of describing complex shapes is further simplified by extracting feature values from nerves that are imbedded in simplicial complexes covering a spatial region.  This is essentially a point-free geometry approach introduced by ~\cite{Peters2016arXivProximalPhysicalGeometry}.

 \begin{figure}[!ht]
	\centering
	\subfigure[Edge 2-spoke]{
	 \includegraphics[width=25mm]{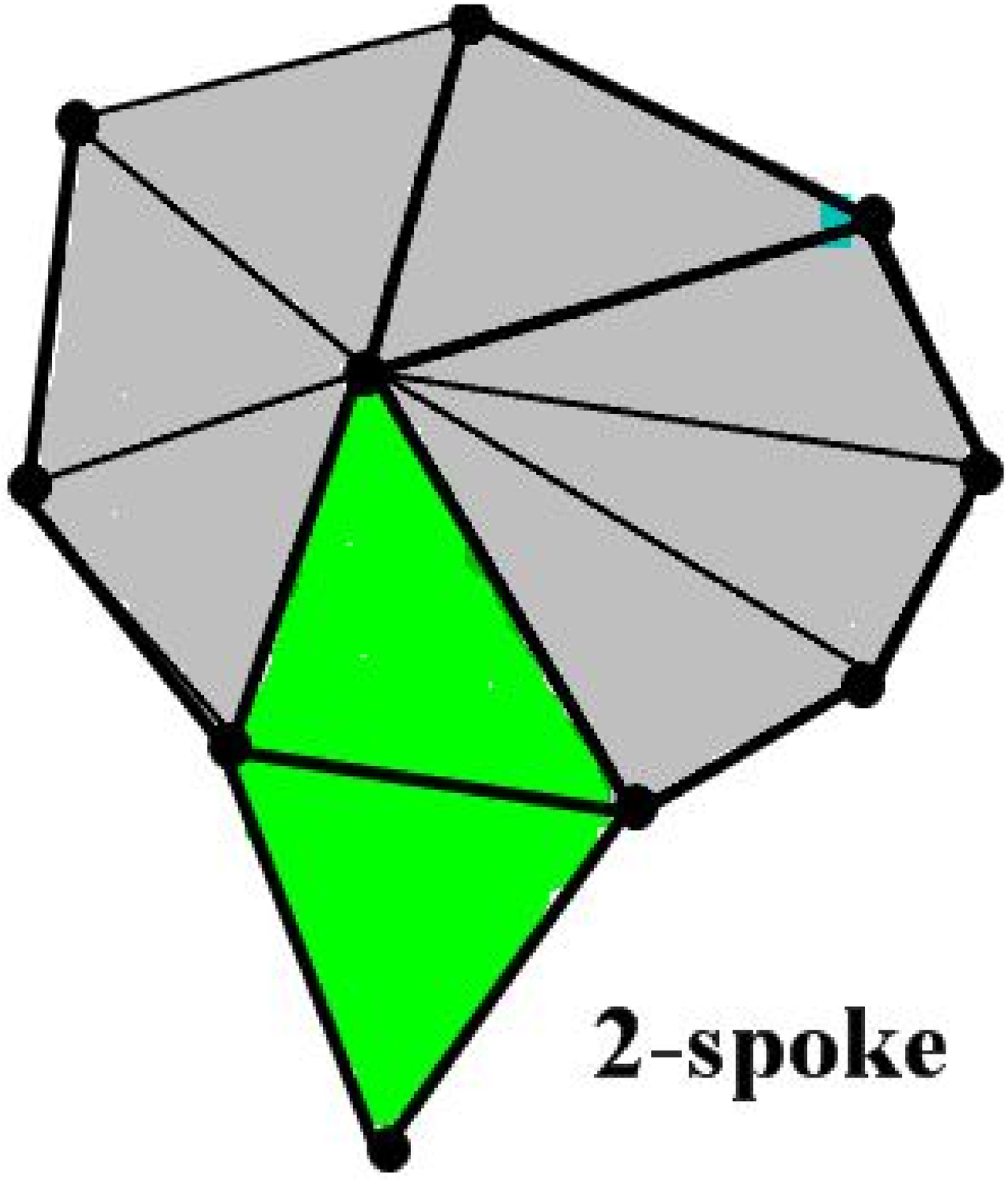}\label{fig:2-spokeEdge}}\hfil
	\subfigure[Edges 2-spokes]{
	 \includegraphics[width=25mm]{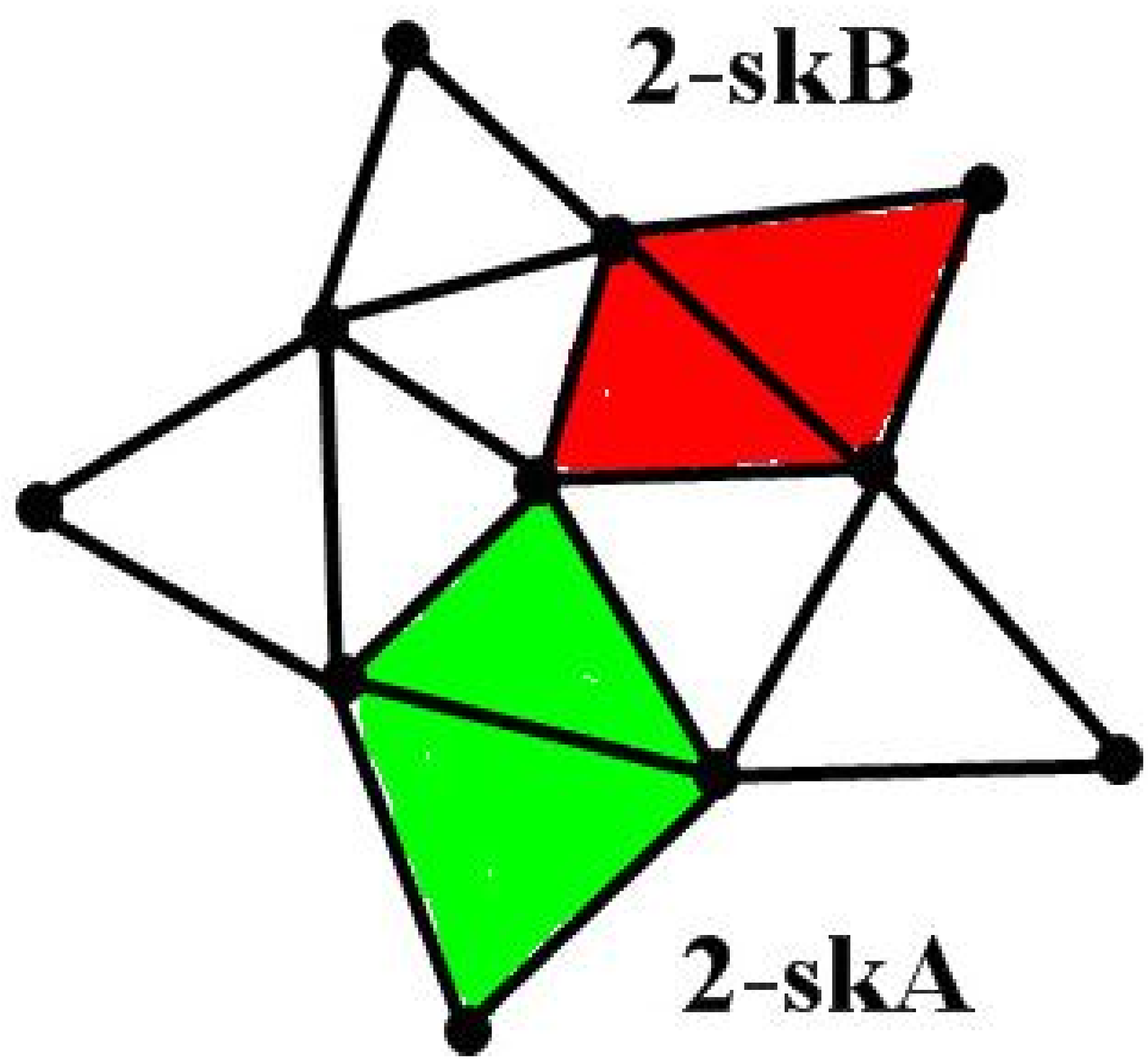}\label{fig:2spokeEdges}}
	\subfigure[Vertex 2-spoke]{
	 \includegraphics[width=25mm]{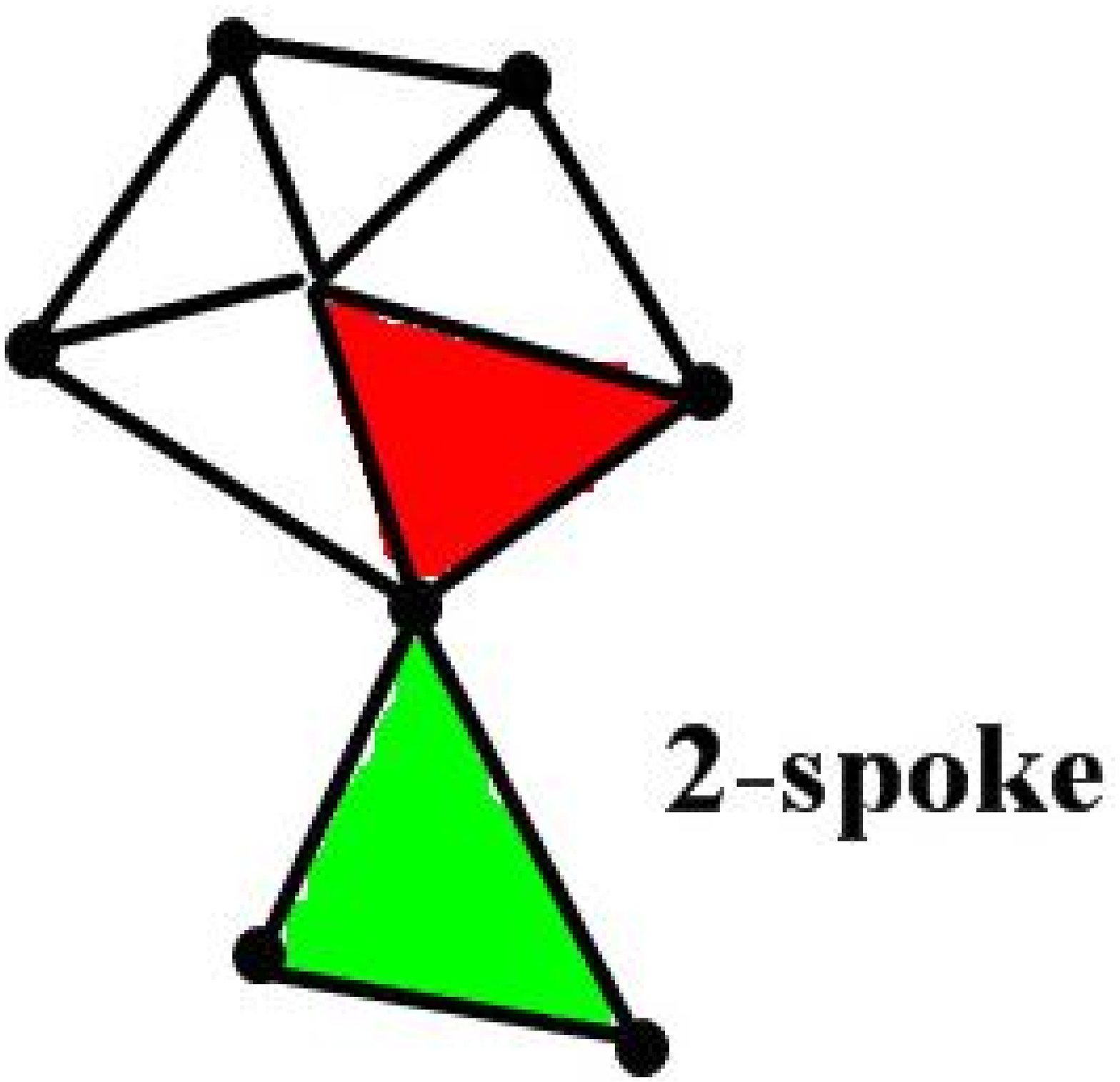}\label{fig:2spokeVertex}}
	\subfigure[Vertex 2-spokes]{
	 \includegraphics[width=25mm]{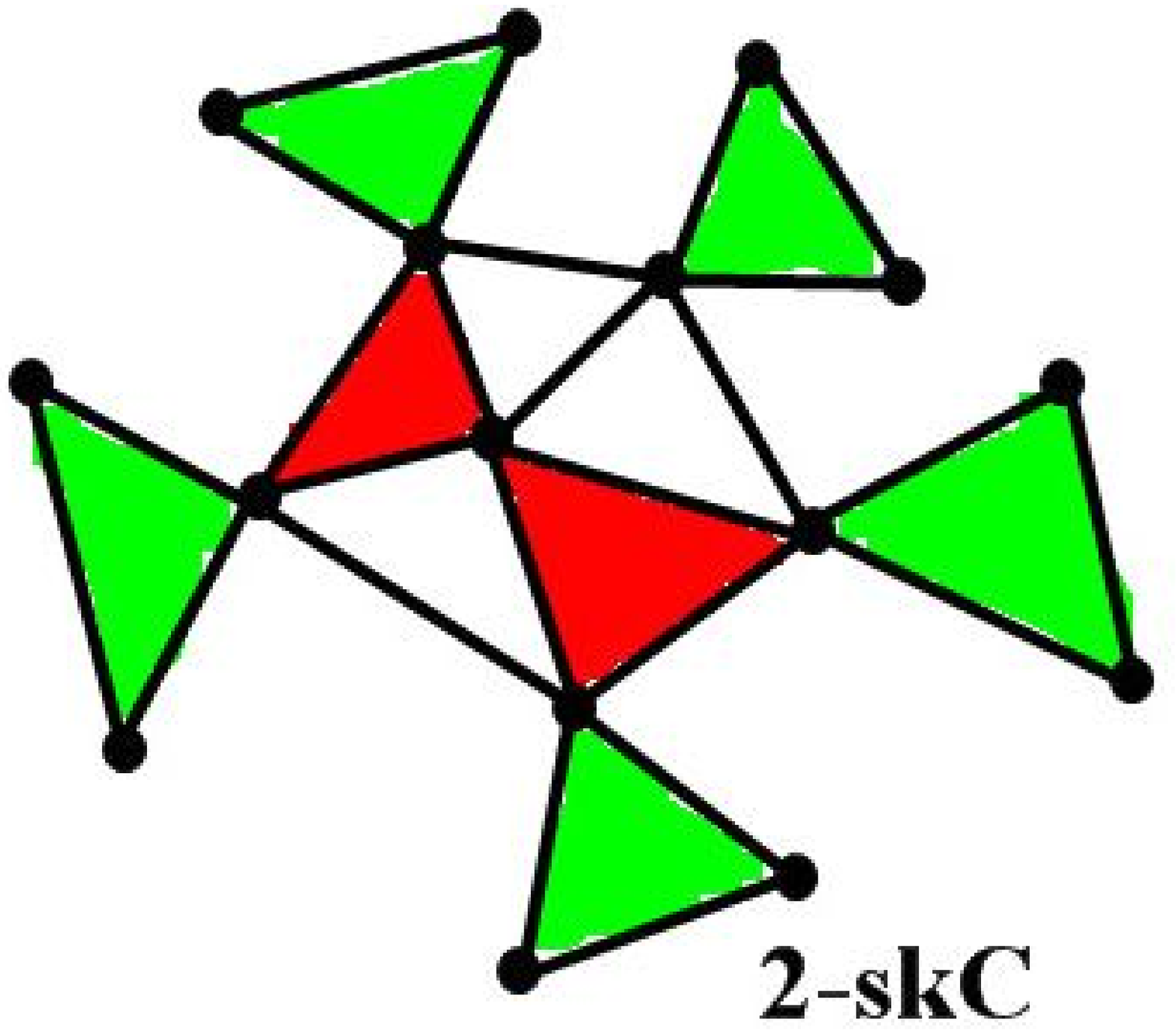}\label{fig:2spokeVertices}}
	\caption{Two forms of nerve 2-spokes}
	\label{fig:2spokes}
\end{figure}

 \begin{figure}[!ht]
	\centering
	\subfigure[Edge 3-spoke]{
	 \includegraphics[width=25mm]{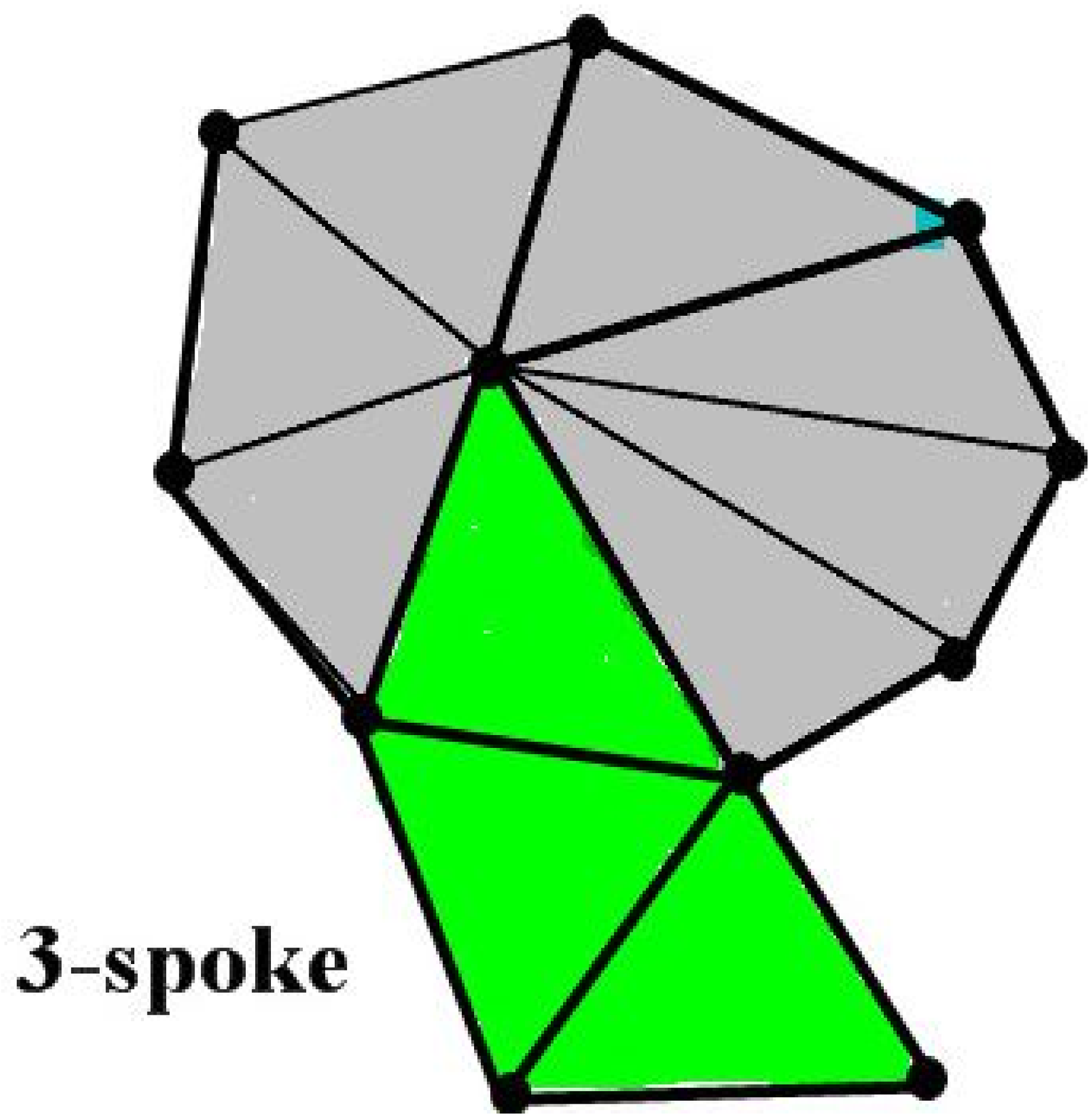}\label{fig:3spokeEdge}}\hfil
	\subfigure[Edges 3-spokes]{
	 \includegraphics[width=25mm]{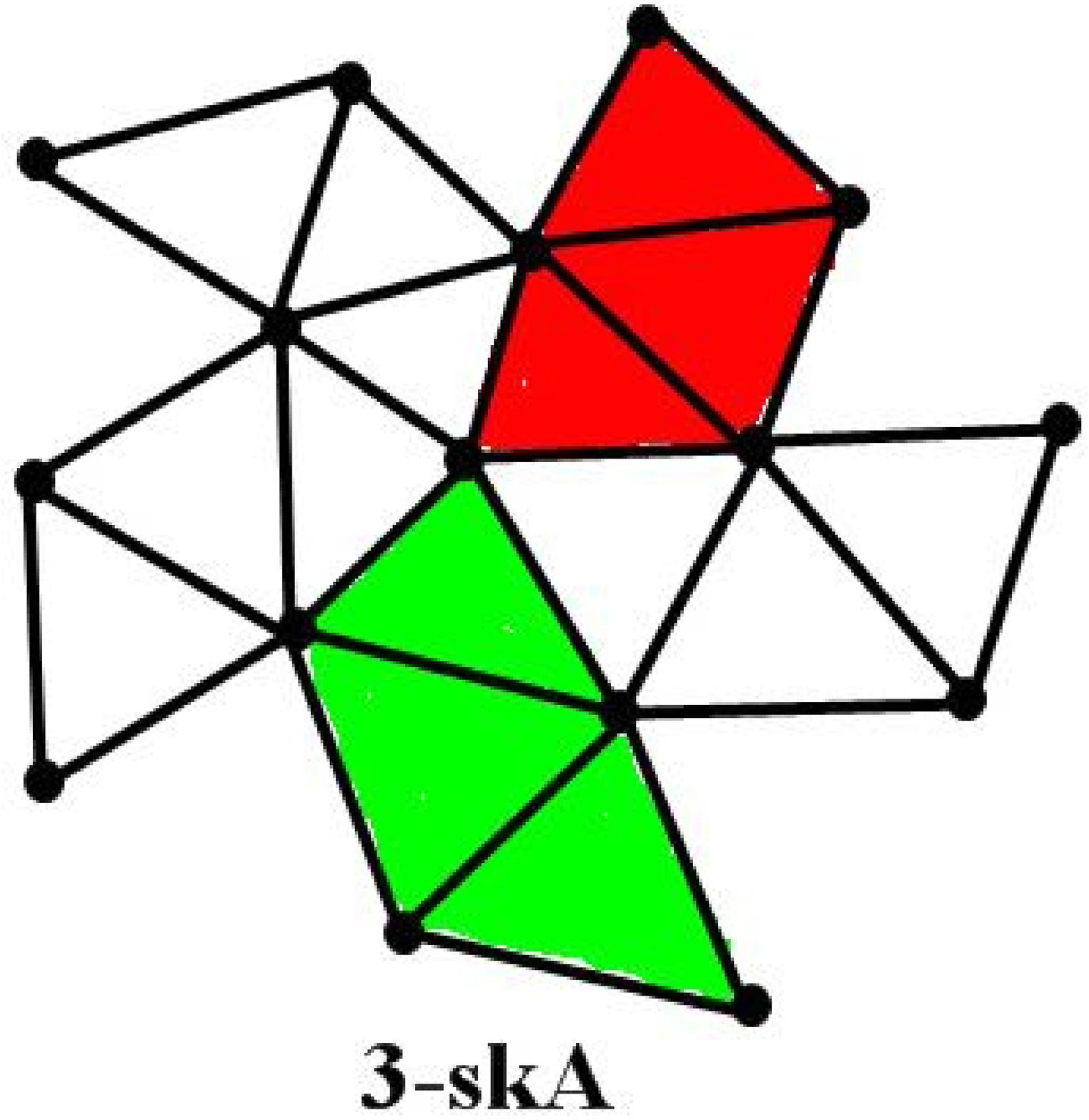}\label{fig:3spokeEdges}}
	\subfigure[Vertex 3-spoke]{
	 \includegraphics[width=25mm]{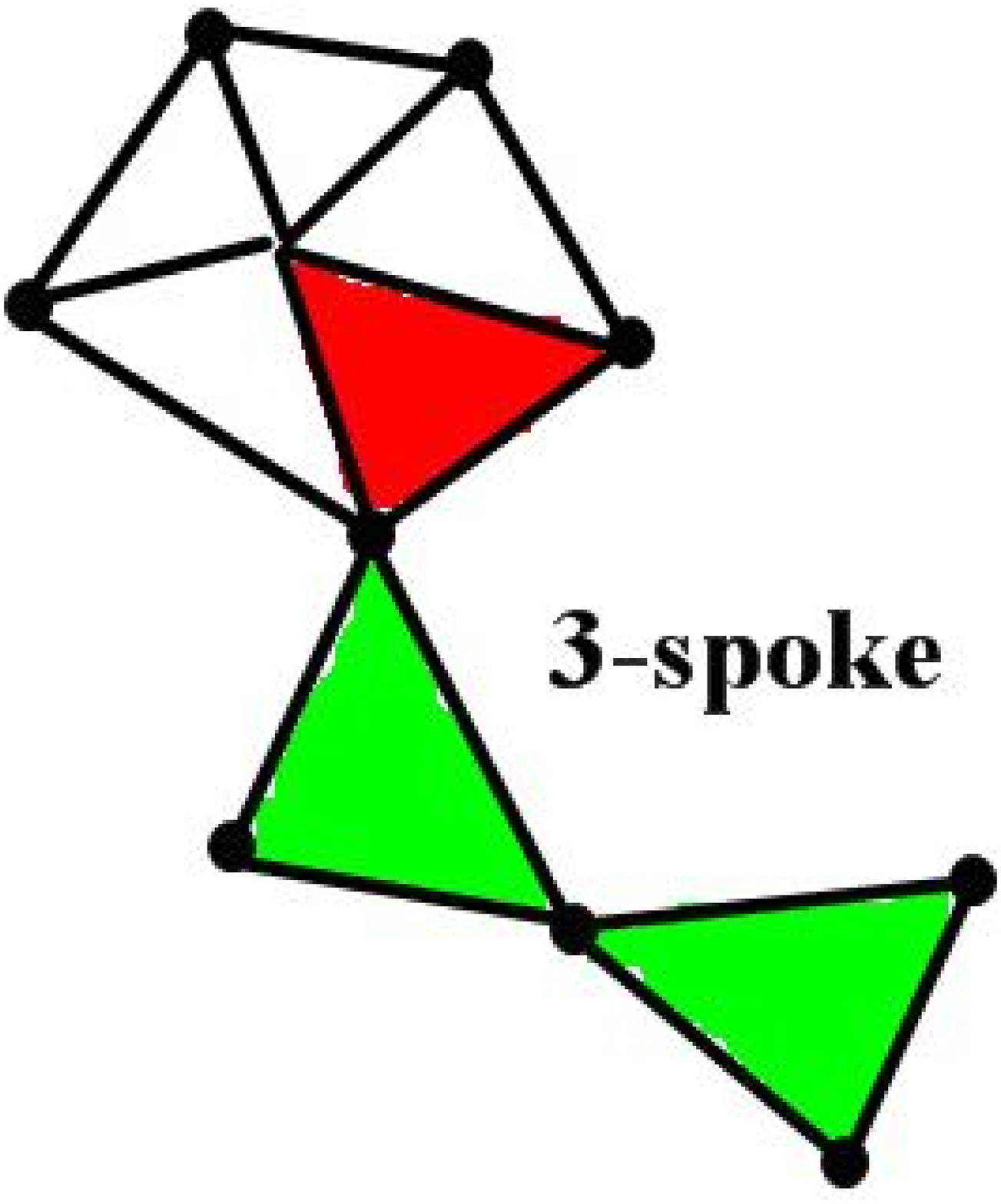}\label{fig:3spokeVertex}}
	\subfigure[Vertex 3-spokes]{
	 \includegraphics[width=25mm]{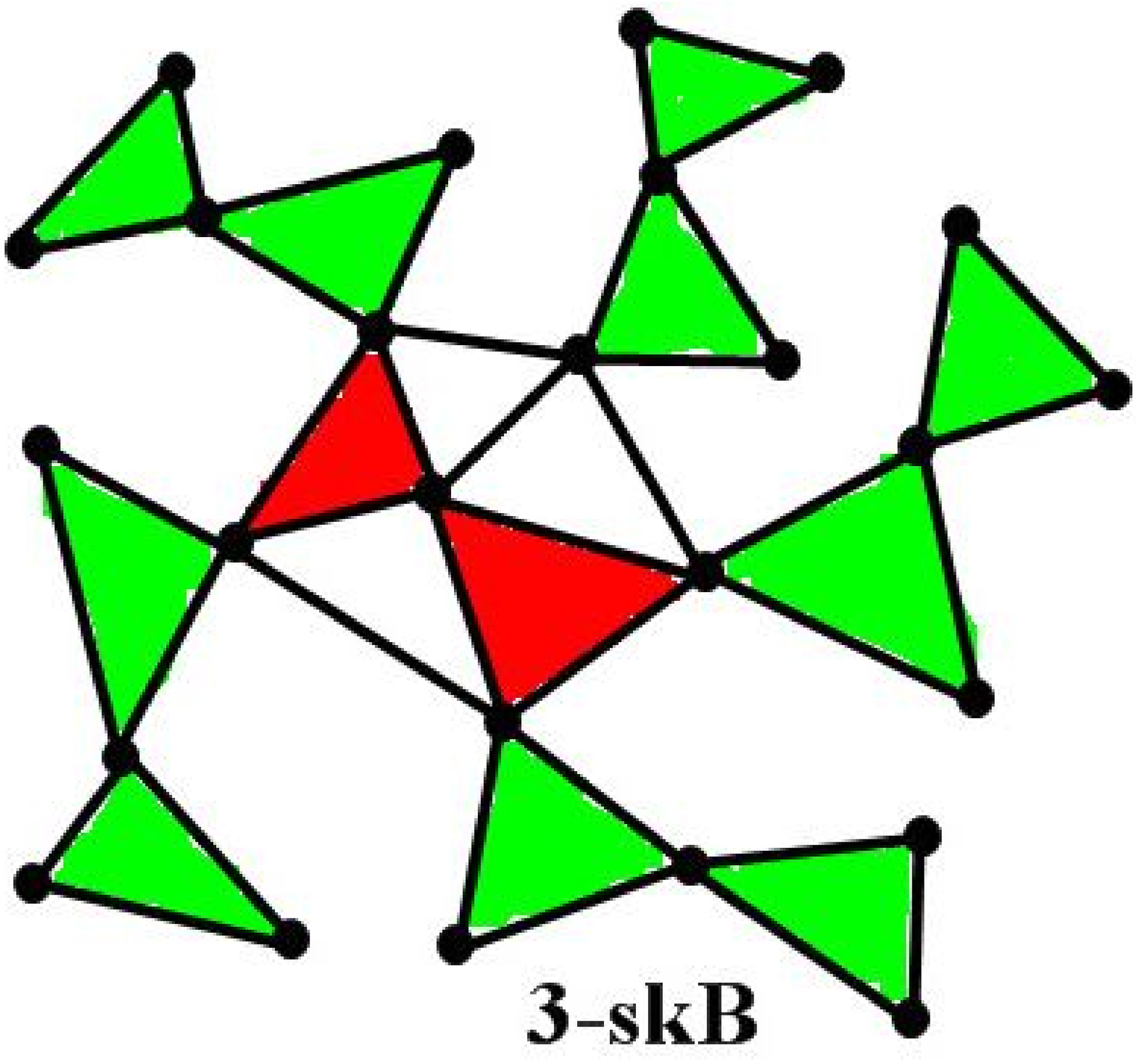}\label{fig:3spokeVertices}}
	\caption{Two forms of nerve 3-spokes}
	\label{fig:3spokes}
\end{figure}

\section{Preliminaries}
This section briefly introduces shape geometry, an extension of nerve complexes in the form of $k$-spokes, $k\geq 1$ and proximities useful in the study nerve complexes.

\subsection{Shape Geometry}
This section briefly introduces structures useful in the study of collections of close filled triangles in the triangulation of planar digital images.  A filled triangle is defined in terms of the boundary and the interior of a set of surface vertices.  A sample algorithm useful in triangulating a digital image is given in Alg.~\ref{alg:ImageGeometry}.$\mbox{}$\\
\vspace{3mm}

\begin{algorithm}[!ht]
\caption{Digital Image Geometry via Triangulation}\label{alg:ImageGeometry}

\SetKwData{Left}{left}
\SetKwData{This}{this}
\SetKwData{Up}{up}
\SetKwFunction{Union}{Union}
\SetKwFunction{FindCompress}{FindCompress}
\SetKwInOut{Input}{Input}
\SetKwInOut{Output}{Output}
\SetKwComment{tcc}{/*}{*/}

\Input{Read digital image $img$.}
\Output{Triangulation Mesh $\mathscr{M}$ covering an image.}
\emph{$img \longmapsto SelectedVertices$}\;
\emph{$S \leftarrow SelectedVertices$}\;
/* $S$ contains vertex coordinates used to triangulate $img$. */ \;
\emph{Continue  $\leftarrow$ True; Vertices  $\leftarrow$ emptyset}\;
\While {($S\neq\emptyset\ \&\ Continue = True$)}{
  /* Select neighbouring vertices $\left\{p,q,r\right\}\subset S$. */ \;
  \emph{$Vertices \assign Vertices\cup \left\{p,q,r\right\}$}\;
	/* $\fil \Delta(pqr)$ = intersection of three closed half spaces. */ \;
	\emph{$\mathscr{M}\ \assign\ \mathscr{M}\cup \fil \Delta(pqr)$}\;
	\emph{$S\ \assign\ S\setminus Vertices$}\;
	\eIf{$S\neq \emptyset$}{
	/* Continue */ }{
	\emph{Continue $\leftarrow$ False}\;
	} 
} 
\emph{$\mathscr{M}\longmapsto\ img$} \;
/* Use $\mathscr{M}$ to gain information about image shape geometry. */ \;
\end{algorithm}

Let $A\ \near\ B$ indicate that the nonempty sets $A$ and $B$ are close to each other in a space $X$ equipped with the proximity $\delta$ (for the details, see Section~\ref{sec:proximities}).  The \emph{\bf boundary} of $A$ (denoted $\mbox{bdy}A$) is the set of all points that are close to $A$ and close to the complement of $A$~\cite[\S 2.7, p. 62]{Naimpally2013}. The \emph{\bf closure} of $A$ (denoted by $\cl A$) is defined by
\[
\cl A = \left\{x\in X: x\ \near\ A\right\}\ \mbox{(Closure of $A$)}.
\]
 An important structure is the \emph{\bf interior} of $A$ (denoted $\mbox{int}A$), defined by $\mbox{int}A = \mbox{cl}A - \mbox{bdy}A$.  Let $p,q,r$ be points in the space $X$.  A \emph{\bf filled triangle} (denoted by $\fil \Delta(pqr)$) is defined by
\[
\fil \Delta(pqr) = \Int \Delta(pqr)\ \cup\ \bdy\ \Delta(pqr)\ \mbox{(filled triangle)}.
\]

When it is clear from the context that simplex triangles are referenced, we write $\Delta(pqr)$ or $\Delta A$ or simply $\Delta$, instead of $\fil \Delta(pqr)$.  Since image object shapes tend to be irregular, the geometry of 2-simplexes covering of an image gives a precise view of the shapes of image objects.
From the known properties of 2-simplexes ({\em e.g.}, interior angles, perimeter, area, lengths of sides), object shape interiors and contours covered by 2-simplexes can be described in a very accurate fashion.

\subsection{Nerve Spokes}
A nerve spoke extends outward from the nucleus of a nerve $\Nrv K$, giving rise to 1-spokes (2-simplexes within a nerve $\Nrv K$), 2-spokes (a 2-simlex that has either a vertex or an edge in common with a 1-spoke in $\Nrv K$), 3-spokes (a 2-simlex that has either a vertex or an edge in common with a 2-spoke in $\Nrv K$),..., and $n$-spokes (a $2$-simlex that has either a vertex or an edge in common with an $n-1$-spoke in $\Nrv K$). 

\begin{example}
Two forms of 2-spokes are shown in Fig.~\ref{fig:2spokes}.  Let $\Nrv K$ be a planar nerve.  An edge-based 2-spoke consists of a nerve $\Nrv K$ 1-spoke that has an edge in common with a non-nerve 1-spoke (see, {\em e.g.}, Fig.~\ref{fig:2-spokeEdge}).  A complete collection of edge-based 2-spokes on $\Nrv K$ is shown in Fig.~\ref{fig:2spokeEdges}.

An vertex-based 2-spoke consists of a nerve $\Nrv K$ 1-spoke that has a vertex in common with a non-nerve 1-spoke (see, {\em e.g.}, Fig.~\ref{fig:2spokeVertex}).  A complete collection of vertex-based 2-spokes on $\Nrv K$ is shown in Fig.~\ref{fig:2spokeVertices}.
\qquad \textcolor{blue}{\Squaresteel}
\end{example}

Edge-based and vertex-based 3-spokes are also possible.

\begin{example}
Two forms of 3-spokes are shown in Fig.~\ref{fig:3spokes}.  An edge-based 3-spoke consists of a nerve $\Nrv K$ 2-spoke that has an edge in common with a non-nerve 1-spoke (see, {\em e.g.}, Fig.~\ref{fig:3spokeEdge}).  A complete collection of edge-based 3-spokes on $\Nrv K$ is shown in Fig.~\ref{fig:3spokeEdges}.

An vertex-based 3-spoke consists of a nerve $\Nrv K$ 2-spoke that has a vertex in common with a non-nerve 1-spoke (see, {\em e.g.}, Fig.~\ref{fig:3spokeVertex}).  A complete collection of vertex-based 3-spokes on $\Nrv K$ is shown in Fig.~\ref{fig:3spokeVertices}.
\qquad \textcolor{blue}{\Squaresteel}
\end{example}

\setlength{\intextsep}{0pt}

\begin{wrapfigure}[16]{R}{0.40\textwidth}
\begin{minipage}{3.8 cm}
\centering
\includegraphics[width=45mm]{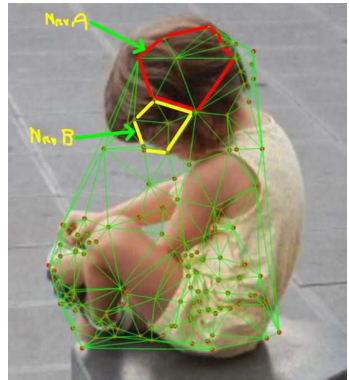}
\caption[]{\footnotesize\bf Overlapping Nerves}
\label{fig:overlap}
\end{minipage}
\end{wrapfigure}
$\mbox{}$\\
\vspace{5mm}

\subsection{Descriptions and Proximities}\label{sec:proximities}
This section briefly introduces two basic types of proximities, namely, traditional \emph{spatial proximity} and the more recent \emph{descriptive proximity} in the study of computational proximity~\cite{Peters2016CP}.  
Nonempty sets that have \emph{\bf spatial proximity} are close to each other, either asymptotically or with common points.  Sets with points in common are strongly proximal. Nonempty sets that have \emph{\bf descriptive proximity} are close, provided the sets contain one or more elements that have matching descriptions.  A commonplace example of descriptive proximity is a pair of paintings that have matching parts such as matching facial characteristics, matching eye, hair, skin colour, or matching nose, mouth, ear shape. Each of these proximities has a strong form.  A \emph{\bf strong proximity} embodies a special form of tightly twisted nextness of nonempty sets.  In simple terms, this means sets that share elements, have strong proximity.   
$\mbox{}$\\

\begin{figure}[!ht]
\centering
\includegraphics[width=45mm]{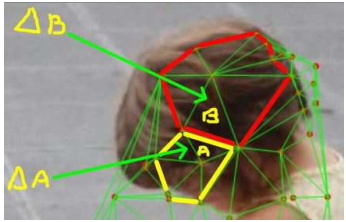}
\caption[]{Strongly near filled triangles}
\label{fig:snTriangles}
\end{figure}

\begin{example}
From Fig.~\ref{fig:overlap}, the pair of nerves $\Nrv A, \Nrv B$ exhibit strong proximity, since there is a 2-spoke that overlaps the nerves.  To see this, consider triangles $\Delta A, \Delta B$ in Fig.~\ref{fig:snTriangles} and let $\sk E$be a 2-spoke in $\Nrv A$ defined by $\sk E = \Delta A\cup \Delta B$.  Similarly, $\Nrv B$ has a 2-spoke $\sk H$, also defined by $\sk H = \Delta A\cup \Delta B$.  In other words, nerves $\Nrv A, \Nrv B$ overlap due to their 2-spokes $\sk E, \sk H$. In effect, $\Delta A, \Delta B$ are strongly near, since these triangles belong to overlapping 2-spokes.
\qquad \textcolor{blue}{\Squaresteel}
\end{example}

\emph{\bf Proximities} are nearness relations.  In other words, a \emph{proximity} between nonempty sets is a mathematical expression that specifies the closeness of the sets.   A \emph{\bf proximity space} results from endowing a nonempty set with one or more proximities.   Typically, a proximity space is endowed with a common proximity such as the proximities from \u Cech~\cite{Cech1966}, Efremovi\u c~\cite{Efremovic1952}, Lodato~\cite{Lodato1962}, and Wallman~\cite{Wallman1938}, or the more recent descriptive proximity~\cite{Peters2007FINearSets,Peters2007AMS,Peters2013mcsintro}.

A pair of nonempty sets in a proximity space are \emph{near} (\emph{close to each other}), provided the sets have one or more points in common or each set contains one or more points that are sufficiently close to each other.  Let $X$ be a nonempty set, $A,B,C\subset X$. E. \u{C}ech~\cite{Cech1966} introduced axioms for the simplest form of proximity $\delta_C$, which satisfies\\
\vspace{3mm}

\noindent \fbox{\bf \u Cech Proximity Axioms~\cite[\S 2.5, p. 439]{Cech1966}}
\begin{description}
\item[{\rm\bf (P1)}] $\emptyset \not\delta A, \forall A \subset X $.
\item[{\rm\bf (P2)}] $A\ \delta\ B \Leftrightarrow B \delta A$.
\item[{\rm\bf (P3)}] $A\ \cap\ B \neq \emptyset \Rightarrow A \near B$.
\item[{\rm\bf (P4)}] $A\ \delta\ (B \cup C) \Leftrightarrow A\ \delta\ B $ or $A\ \delta\ C$. \qquad \textcolor{blue}{$\blacksquare$}
\end{description}

The Lodato proximity $\delta_L$ satisfies the \u Cech proximity axioms and axiom (P5).\\
\vspace{3mm}

\noindent \fbox{\bf Lodato Proximity Axiom~\cite{Lodato1962}}
\begin{description}
\item[{\rm\bf (P5)}]  $A\ \delta_L\ B$ and $\{b\}\ \delta_L\ C$ for each $b \in B \ \Rightarrow A\ \delta_L\ C$.
\qquad \textcolor{blue}{$\blacksquare$}
\end{description}

\noindent We can associate a topology with the space $(X, \delta)$ by considering as closed sets those sets that coincide with their own closure.

Nonempty sets $A,B$ in a topological space $X$ equipped with the proximity $\sn$ are \emph{strongly near} [\emph{strongly contacted}] (denoted $A\ \sn\ B$), provided the sets have at least one point in common.   The strong contact relation $\sn$ was introduced in~\cite{Peters2015JangjeonMSstrongProximity} and axiomatized in~\cite{PetersGuadagni2015stronglyNear},~\cite[\S 6 Appendix]{Guadagni2015thesis}.\\

\noindent \fbox{\bf Strong Proximity~\cite[\S 1.2]{Peters2015AMSJstrongProximity} (see, also,~\cite[\S 1.5]{Peters2016CP},~\cite{Peters2015JangjeonMSstrongProximity,PetersGaudagni2015arXivVoronoiManifolds}).}\\

Let $X$ be a topological space, $A, B, C \subset X$ and $x \in X$.  The relation $\sn$ on the family of subsets $2^X$ is a \emph{strong proximity}, provided it satisfies the following axioms.

\begin{description}
\item[{\rm\bf (snN0)}] $\emptyset\ \sfar\ A, \forall A \subset X $, and \ $X\ \sn\ A, \forall A \subset X$.
\item[{\rm\bf (snN1)}] $A \sn B \Leftrightarrow B \sn A$.
\item[{\rm\bf (snN2)}] $A\ \sn\ B$ implies $A\ \cap\ B\neq \emptyset$. 
\item[{\rm\bf (snN3)}] If $\{B_i\}_{i \in I}$ is an arbitrary family of subsets of $X$ and  $A \sn B_{i^*}$ for some $i^* \in I \ $ such that $\Int(B_{i^*})\neq \emptyset$, then $  \ A \sn (\bigcup_{i \in I} B_i)$ 
\item[{\rm\bf (snN4)}]  $\mbox{int}A\ \cap\ \mbox{int} B \neq \emptyset \Rightarrow A\ \sn\ B$.  
\qquad \textcolor{blue}{$\blacksquare$}
\end{description}

\noindent When we write $A\ \sn\ B$, we read $A$ is \emph{strongly near} $B$ ($A$ \emph{strongly contacts} $B$).  The notation $A\ \notfar\ B$ reads $A$ is not strongly near $B$ ($A$ does not \emph{strongly contact} $B$). For each \emph{strong proximity} (\emph{strong contact}), we assume the following relations:
\begin{description}
\item[{\rm\bf (snN5)}] $x \in \Int (A) \Rightarrow x\ \sn\ A$ 
\item[{\rm\bf (snN6)}] $\{x\}\ \sn \{y\}\ \Leftrightarrow x=y$  \qquad \textcolor{blue}{$\blacksquare$} 
\end{description}

For strong proximity of the nonempty intersection of interiors, we have that $A \sn B \Leftrightarrow \Int A \cap \Int B \neq \emptyset$ or either $A$ or $B$ is equal to $X$, provided $A$ and $B$ are not singletons; if $A = \{x\}$, then $x \in \Int(B)$, and if $B$ too is a singleton, then $x=y$. It turns out that if $A \subset X$ is an open set, then each point that belongs to $A$ is strongly near $A$.  The bottom line is that strongly near sets always share points, which is another way of saying that sets with strong contact have nonempty intersection.   Let $\near$ denote a traditional proximity relation~\cite{Naimpally1970}.

\begin{proposition}\label{prop:2spoke}
Let $\Nrv A, \Nrv B$ be nerve complexes in a triangulated space $X$.  $\Nrv A\ \sn\ \Nrv B$, if and only if
2-spoke $\sk E\in \Nrv A\ \cap\ \Nrv B$ for some 2-spoke common to the pair of nerves.
\end{proposition}

\begin{corollary}
Let $\Nrv A, \Nrv B$ be nerve complexes in a triangulated space $X$.  
A 3-spoke $\sk H\in \Nrv A\ \cup\ \Nrv B$ for some 3-spoke common to the pair of nerves.
\end{corollary}
\begin{proof}
Immediate from Prop.~\ref{prop:2spoke} and the definition of a 3-spoke.
\end{proof}

\subsection{Descriptive Proximities}
In the run-up to a close look at extracting features of triangulated image objects, we first consider descriptive proximities, fully covered in~\cite{DiConcilio2016arXivDescriptiveProximities} and briefly introduced, here.  There are two basic types of \emph{object features}, namely, \emph{object characteristic} and \emph{object location}.  For example, an object characteristic of a picture point is colour.  
Descriptive proximities resulted from the introduction of the descriptive intersection pairs of nonempty sets. \\
\vspace{3mm} 

\noindent \fbox{\bf Descriptive Intersection~\cite{Peters2013mcsintro} and~\cite[\S 4.3, p. 84]{Naimpally2013}.}
\begin{description}
\item[{\rm\bf ($\boldsymbol{\Phi}$)}] $\Phi(A) = \left\{\Phi(x)\in\mathbb{R}^n: x\in A\right\}$, set of feature vectors.
\item[{\rm\bf ($\boldsymbol{\dcap}$)}]  $A\ \dcap\ B = \left\{x\in A\cup B: \Phi(x)\in \Phi(A) \& \in \Phi(x)\in \Phi(B)\right\}$.
\qquad \textcolor{blue}{$\blacksquare$}
\end{description}
$\mbox{}$\\
\vspace{3mm}

\setlength{\intextsep}{0pt}

\begin{wrapfigure}[15]{R}{0.55\textwidth}
\begin{minipage}{4.2 cm}
\centering
\includegraphics[width=65mm]{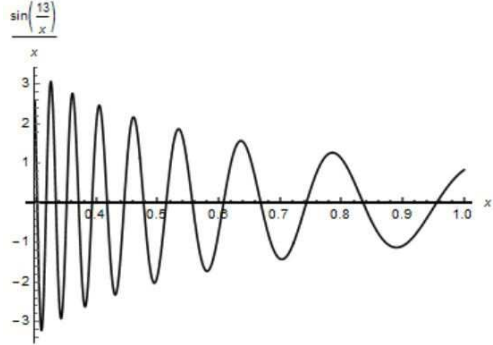}
\caption[]{\footnotesize\bf Strongly Near}
\label{fig:stronglyNearSets}
\end{minipage}
\end{wrapfigure}
$\mbox{}$\\
\vspace{5mm}

The descriptive proximity $\delta_{\Phi}$ was introduced in~\cite{Peters2007FINearSets,Peters2007AMS,Peters2013mcsintro}.  Let $\Phi(x)$ be a feature vector for $x\in X$, a nonempty set of non-abstract points such as picture points.  $A\ \delta_{\Phi}\ B$ reads $A$ is descriptively near $B$, provided $\Phi(x) = \Phi(y)$ for at least one pair of points, $x\in A, y\in B$.  The proximity $\delta$ in the \u{C}ech, Efremovi\u c, and Wallman proximities is replaced by $\dnear$.  Then swapping out $\near$ with $\dnear$ in each of the Lodato axioms defines a descriptive Lodato proximity that satisfies the following axioms.\\
\vspace{3mm}

\noindent \fbox{\bf Descriptive Lodato Axioms~\cite[\S 4.15.2]{Peters2013springer}}
\begin{description}
\item[{\rm\bf (dP0)}] $\emptyset\ \dfar\ A, \forall A \subset X $.
\item[{\rm\bf (dP1)}] $A\ \dnear\ B \Leftrightarrow B\ \dnear\ A$.
\item[{\rm\bf (dP2)}] $A\ \dcap\ B \neq \emptyset \Rightarrow\ A\ \dnear\ B$.
\item[{\rm\bf (dP3)}] $A\ \dnear\ (B \cup C) \Leftrightarrow A\ \dnear\ B $ or $A\ \dnear\ C$.
\item[{\rm\bf (dP4)}] $A\ \dnear\ B$ and $\{b\}\ \dnear\ C$ for each $b \in B \ \Rightarrow A\ \dnear\ C$. \qquad \textcolor{blue}{$\blacksquare$}
\end{description}

Nonempty sets $A,B$ in a proximity space $X$ are \emph{strongly near} (denoted $A\ \sn\ B$), provided the sets share points.  Strong proximity $\sn$ was introduced in~\cite[\S 2]{Peters2015AMSJstrongProximity} and completely axiomatized in~\cite{PetersGuadagni2015stronglyNear} (see, also,~\cite[\S 6 Appendix]{Guadagni2015thesis}).

\begin{proposition}\label{prop:dnear}
Let $\left(X,\dnear\right)$ be a descriptive proximity space, $A,B\subset X$.  Then $A\ \dnear\ B \Rightarrow A\ \dcap\ B\neq \emptyset$.
\end{proposition}
\begin{proof}
$A\ \dnear\ B \Leftrightarrow$ there is at least one $x\in A, y\in B$ such that $\Phi(x)=\Phi(y)$ (by definition of $A\ \dnear\ B$).  Hence, $A\ \dcap\ B\neq \emptyset$.
\end{proof}

Next, consider a proximal form of a Sz\'{a}z relator~\cite{Szaz1987}.  A \emph{proximal relator} $\mathscr{R}$ is a set of relations on a nonempty set $X$~\cite{Peters2016relator}.  The pair $\left(X,\mathscr{R}\right)$ is a proximal relator space.  The connection between $\sn$ and $\near$ is summarized in Prop.~\ref{thm:sn-implies-near}.

\begin{lemma}\label{thm:sn-implies-near}
Let $\left(X,\left\{\near,\dnear,\sn\right\}\right)$ be a proximal relator space, $A,B\subset X$.  Then 
\begin{compactenum}[1$^o$]
\item $A\ \sn\ B \Rightarrow A\ \near\ B$.
\item $A\ \sn\ B \Rightarrow A\ \dnear\ B$.
\end{compactenum}
\end{lemma}
\begin{proof}$\mbox{}$\\
1$^o$: From Axiom (snN2), $A\ \sn\ B$ implies $A\ \cap\ B\neq \emptyset$, which implies $A\ \near\ B$ (from Lodato Axiom (P2)).\\
2$^o$: From 1$^o$, there are $x\in A, y\in B$\ common to $A$ and $B$.  Hence, $\Phi(x) = \Phi(y)$, which implies $A\ \dcap\ B\neq \emptyset$.  Then, from the descriptive Lodato Axiom (dP2), $A\ \dcap\ B \neq \emptyset \Rightarrow\ A\ \dnear\ B$. This gives the desired result.
\end{proof}

\begin{theorem}\label{thm:spoke}
Let $\left(X,\left\{\dnear,\sn\right\}\right)$ be a proximal relator triangulated space, $\Nrv A,\Nrv B\subset 2^X$.  Then
\begin{compactenum}[1$^o$]
\item $\Nrv A\ \sn\ \Nrv B$ implies $\Nrv A\ \dnear\ \Nrv B$.
\item A 2-spoke $\sk A\ \in\ \Nrv A\cap \Nrv B$ implies if $\sk A\ \in\ \Nrv A\ \dcap\ \Nrv B$.
\item A 2-spoke $\sk A\ \in\ \Nrv A\cap \Nrv B$ implies $\Nrv A\ \dnear\ \Nrv B$.
\end{compactenum}
\end{theorem}
\begin{proof}$\mbox{}$\\
1$^o$: Immediate from part 2$^o$ Lemma~\ref{thm:sn-implies-near}.\\
2$^o$: From Prop.~\ref{prop:2spoke}, $\sk A\ \in\ \Nrv A\cap \Nrv B$, if and only if $\Nrv A\ \sn\ \Nrv B$.  Consequently, there are members of the 2-spoke $\sk A$ common to $\Nrv A,\Nrv B$, which have the same description. Hence, $\sk A\ \in\ \Nrv A\ \dcap\ \Nrv B$.\\
3$^o$: Immediate from 2$^o$ and Lemma~\ref{thm:sn-implies-near}.
\end{proof}

\begin{example}
Let $X$ be a topological space endowed with the strong proximity $\sn$ and $A = \left\{(x,0): 0.1\leq x\leq 1\right\}$,$B = \left\{(x,\frac{1}{x}sin(13/x)):0.1\leq x\leq 1\right\}$.  In this case, $A,B$ represented by
Fig.~\ref{fig:stronglyNearSets} are strongly near sets with many points in common.  
\qquad \textcolor{blue}{$\blacksquare$}
\end{example}

The descriptive strong proximity $\snd$ is the descriptive counterpart of $\sn$. 

\begin{definition}\label{def:snd}
Let $X$ be a topological space, $A, B, C \subset X$ and $x \in X$.  The relation $\snd$ on the family of subsets $2^X$ is a \emph{descriptive strong Lodato proximity}, provided it satisfies the following axioms.
\vspace{2mm}

\noindent \fbox{\bf Descriptive Strong Lodato proximity~\cite[\S 4.15.2]{Peters2013springer}}
\begin{description}
\item[{\rm\bf (dsnN0)}] $\emptyset\ {\sdfar}\ A, \forall A \subset X $, and \ $X\ \snd\ A, \forall A \subset X$
\item[{\rm\bf (dsnN1)}] $A\ \snd\ B \Leftrightarrow B\ \snd\ A$
\item[{\rm\bf (dsnN2)}] $A\ \snd\ B \Rightarrow\ A\ \dcap\ B \neq \emptyset$
\item[{\rm\bf (dsnN3)}] If $\{B_i\}_{i \in I}$ is an arbitrary family of subsets of $X$ and  $A\ \snd\ B_{i^*}$ for some $i^* \in I \ $ such that $\Int(B_{i^*})\neq \emptyset$, then $A\ \snd\ (\bigcup_{i \in I} B_i)$
\item[{\rm\bf (dsnN4)}] $\Int A\ \dcap\ \Int B \neq \emptyset \Rightarrow A\ \snd\ B$
\qquad \textcolor{blue}{$\blacksquare$}
\end{description}
\end{definition}

\noindent When we write $A\ \snd\ B$, we read $A$ is \emph{descriptively strongly near} $B$.  The notation $A\ \sdfar\ B$ reads $A$ is not descriptively strongly near $B$.
For each \emph{descriptive strong proximity}, we assume the following relations:
\begin{description}
\item[(dsnN5)] $\Phi(x) \in \Phi(\Int (A)) \Rightarrow x\ \snd\ A$ 
\item[(dsnN6)] $\{x\}\ \snd\ \{y\} \Leftrightarrow \Phi(x)=\Phi(y)$  \qquad \textcolor{blue}{$\blacksquare$} 
\end{description}

So, for example, if we take the strong proximity related to non-empty intersection of interiors, we have that $A\ \snd\ B \Leftrightarrow \Int A\ \dcap\ \Int B \neq \emptyset$ or either $A$ or $B$ is equal to $X$, provided $A$ and $B$ are not singletons; if $A = \{x\}$, then $\Phi(x) \in \Phi(\Int(B))$, and if $B$ is also a singleton, then $\Phi(x)=\Phi(y)$.

\begin{example}\label{ex:colourTopology} {\bf Descriptive Strong Proximity}.\\
Let $X$ be a triangulated space of picture points represented in Fig.~\ref{fig:snTriangles} with red, brown or yellow colors and let $\Phi: X \rightarrow  \mathbb{R}^n $ be a description of $X$ representing the color of a picture point, where $0$ stands for red (r), $1$ for brown (b) and $2$ for yellow (y). Suppose the range is endowed with the topology given by $\tau= \{ \emptyset,\{ r,  b\}, \{r,b,y\} \}$.
Then  $\Delta A \ \snd\ \Delta B$, since $\Int \Delta A\ \dcap\ \Int \Delta B\neq \emptyset$, {\em i.e.}, points in the interior of simplexes $\Delta A, \Delta B$ have matching colours.  
\qquad \textcolor{blue}{$\blacksquare$}
\end{example}

\begin{example}\label{ex:MNCnerve} {\bf Nerve Spokes with Descriptive Strong Proximity}.\\
Let $X$ be a planar triangulated region containing a nerve complex $\Nrv K$, equipped with $\snd$.
A pair of spokes $\sk A,\sk B$ in a nerve complex $\Nrv K$ is shown in Fig.~\ref{fig:MNCnerve}.  $\sk A\ \sn\ \sk B$, since $\sk A,\sk B$ have common wiring represented by an overlapping coil.  Hence, from Lemma~\ref{thm:sn-implies-near}, $\sk A\ \dnear\ \sk B$.  From Axiom (dsnN4), $\sk A\ \snd\ \sk B$, since $\Int\ \sk A\ \dcap\ \Int\ \sk B$.
\qquad \textcolor{blue}{\Squaresteel}
\end{example}

\subsection{Closure Nerve Complexes}
This section briefly introduces closure nerves.  Let $\mathscr{F}$ denote a collection of nonempty sets.  An Edelsbrunner-Harer \emph{nerve} of $\mathscr{F}$~\cite[\S III.2, p. 59]{Edelsbrunner2010compTop} (denoted by $\mbox{Nrv} \mathscr{F}$ consists of all nonempty sub-collections whose sets have a common intersection and is defined by
\[
\mbox{Nrv} \mathscr{F} = \left\{X\subseteq \mathscr{F}: \bigcap X \neq \emptyset\right\}.
\]

A natural extension of the basic notion of a nerve arises when we consider adjacent polygons and the closure of a set.  Let $A,B$ be nonempty subset in a topological space $X$.  The expression $A\ \delta\ B$ ($A$ near $B$) holds true for a particular proximity that we choose, provided $A$ and $B$ have nonempty intersection, {\em i.e.}, $A\cap B\neq \emptyset$.  Every nonempty set has a set of points in its interior (denoted $\Int A$) and a set of boundary points (denoted $\bdy A$).  A nonempty set is \emph{open}, provided its boundary set is empty, {\em i.e.}, $\bdy A\neq \emptyset$.  Put another way, a set $A$ is open, provided all points $y\in X$ sufficiently close to $x\in A$ belong to $A$~\cite[\S 1.2]{Bourbaki1966}. A nonempty set is \emph{closed}, provided its boundary set is nonempty.  Notice that a closed set can have an empty interior. 

\begin{example}
Circle, triangles, and or any quadrilaterals are examples of closed sets with either empty or nonempty interiors.  Disks can be either closed or open sets with nonempty interiors. \qquad \textcolor{blue}{\Squaresteel}
\end{example}

An extension of the notion of a nerve in a collection of simplicial complexes $\mathscr{F}$ called a \emph{\bf closure nerve} (denoted $\clNrv \mathscr{F}$), defined by
\[
\clNrv \mathscr{F} = \left\{X\in \mathscr{F}: \bigcap \mbox{cl}X \neq \emptyset\right\}.
\]

\noindent Closure nerves are commonly found in triangulations of digital images.

\begin{example}\label{ex:closureNerve}
Examples of closure nerves are shown in Fig.~\ref{fig:overlap}.  
\qquad \textcolor{blue}{\Squaresteel}
\end{example}

From Example~\ref{ex:closureNerve}, a new form of closure nerve can be derived from the filled triangles in a nerve complex (denoted by $\clNrvt\mathscr{F}$) defined by
\[
\clNrvt \mathscr{F} = \left\{\Delta\in \Nrv{F}: \bigcap \mbox{cl}\Delta \neq \emptyset\right\}.
\]

\setlength{\intextsep}{0pt}
\begin{wrapfigure}[13]{R}{0.50\textwidth}
\begin{minipage}{3.8 cm}
\centering
\begin{pspicture}
(-1.0,-0.5)(5.0,4.2)
\psframe[linecolor=black](-0.5,-0.3)(4.3,4.0)
\providecommand{\PstPolygonNode}{%
 \psdots[dotstyle=o,dotsize=0.08,linecolor=blue,fillcolor=yellow](1;\INode)
 \psline(0.95;\INode)}
\PstPolygon[unit=1.75,PolyNbSides=5,fillstyle=solid,fillcolor=green]
\pscoil[coilarm=.2cm,linecolor=red,coilheight=0.25](-1.8,1.8)(0.3,2.2)
\rput(-3.8,3.7){$\boldsymbol{X}$}
\rput(-2.5,3.5){$\boldsymbol{\Nrv K}$}
\rput(-1.0,2.8){$\boldsymbol{\sk B}$}
\rput(-1.1,1.1){$\boldsymbol{\sk A}$}
\end{pspicture}
\caption[]{\footnotesize $\boldsymbol{\sk A\ \snd\ \sk B}$}
\label{fig:MNCnerve}
\end{minipage}
\end{wrapfigure}

An easy next step is to consider nerve complexes that are descriptively near and descriptively strongly near.  Let $\Nrv A, \Nrv B$ be a pair of nerves and let $\Delta_A\in \Nrv A, \Delta_B\in \Nrv B$.  Then
$
\Nrv A\ \dnear\ \Nrv B,\ \mbox{provided}\ \Delta_A\ \dcap\ \Delta_B\neq \emptyset,\ \mbox{i.e.},
$
$\Delta_A\ \dnear\ \Delta_B$.   Taking this a step further, whenever a region in interior of $\Nrv A$ has a description that matches the description of a region in the interior of $\Nrv B$, the pair of nerves are descriptively strongly near.  Let $\Delta_A\in \Nrv A,\Delta_B\in \Nrv B$.  Then
$
\Nrv A\ \snd\ \Nrv B,\ \mbox{provided}\ \cl \Delta_A\ \dcap\ \cl \Delta_B\neq \emptyset.
$

A descriptive closure nerve complex (denoted by $\dclNrvt\mathscr{F}$) is defined by\\ 
\[
\dclNrvt \mathscr{F} = \left\{\Delta\in \Nrv{F}: \mathop{\bigcap}\limits_{\Phi} \mbox{cl}\Delta \neq \emptyset\right\}.
\]

\begin{lemma}\label{lem:nerveNucleus}
Each closure nerve $\clNrvt \mathscr{F}$ has a nucleus.
\end{lemma}
\begin{proof}
Since all filled triangles in $\clNrvt \mathscr{F}$ have nonempty intersection, the triangles have a common vertex, the nucleus of the nerve.
\end{proof}

\begin{definition}
A pair of filled triangles $\Delta A, \Delta B$ in a triangulated region are \emph{\bf separated triangles}, provided $\Delta A \cap \Delta B = \emptyset$ (the triangles have no points in common) or $\Delta A, \Delta B$ have an edge or a vertex in common and do not have a common nucleus vertex. 
\qquad \textcolor{blue}{\Squaresteel}
\end{definition}

\begin{theorem}\label{lem:nerveFamily}
Let $V$ be a set of vertices, $X$ be a triangulated plane surface covered with 2-simplexes with vertices in $V$.
If $v,v'\in V$ are vertices of separated filled triangles on $X$, then the space has more than one nerve.
\end{theorem}
\begin{proof}
Let $\Delta(v,p,q), \Delta(v',p',q')$ be filled triangles on $X$.  In a nerve, every filled triangle has a pair of edges in common with adjacent triangles and, from Lemma~\ref{lem:nerveNucleus}, the filled triangles in the nerve have a common vertex, namely, the nucleus.  By definition, separated triangles have at most one edge or vertex in common and do not have a common nucleus.  Hence, the separated triangles belong to different nerves.
\end{proof}

\begin{theorem}\label{lem:stronglyNearNerves}
Nerves with a common 2-spoke are strongly near nerves.
\end{theorem}
\begin{proof}
Immediate from the definition of $\sn$.
\end{proof}

\begin{theorem}\label{lem:descriptivelyNearNerves}
Strongly near nerves are strongly descriptively near nerves.
\end{theorem}
\begin{proof}
Let $\Nrv A\ \sn\ \Nrv B$ be strongly near nerves.  Then $\Nrv A, \Nrv B$ have a 2-spoke in common.  Then
$\Int\Nrv A\ \cap\ \Int\Nrv B\neq \emptyset$. Consequently, from Part 2 of Theorem~\ref{thm:spoke}, $\Int\Nrv A\ \dcap\ \Int\Nrv B\neq \emptyset$.  Hence, from Axiom (dsnN4), $\Nrv A\ \snd\ \Nrv B$.
\end{proof}

\begin{theorem}\label{lem:stronglyDescriptivelyNearNerves}
Nerves containing interior regions with matching descriptions are strongly descriptively near nerves.
\end{theorem}
\begin{proof}
Immediate from the definition of $\snd$.
\end{proof}

\section{Main Results}
From a computational topology perspective, homotopy types are introduced in~\cite[\S III.2]{Edelsbrunner1999} and lead to significant results for in the theory of nerve spokes.  

Let $f,g:X\longrightarrow Y$ be two continuous maps.  A \emph{homotopy} between $f$ and $g$ is a continuous map $H:X\times[0,1]\longrightarrow Y$ so that $H(x,0) = f(x)$ and $H(x,1) = g(x)$.  The sets $X$ and $Y$ are \emph{homotopy equivalent}, provided there are continuous maps $f: X\longrightarrow Y$ and $g:Y\longrightarrow X$ such that $g\circ f \simeq \mbox{id}_X$ and $f\circ g \simeq \mbox{id}_Y$.  This yields an equivalence relation $X\simeq Y$.  In addition, $X$ and $Y$ have the same \emph{homotopy type}, provided $X$ and $Y$ are homotopy equivalent.  

Let $\mathscr{F}$ be a finite collection of nerve complexes that cover a space $X$, endowed with the strong proximity $\sn$.  Let $N$ be the nucleus of nerve $\Nrv K$, $K$ a collection of 1-spokes that have in $N$ in common.   Then $\Nrv K$ is defined by
\[
\Nrv K = \left\{\sk A\in K: \sk A\ \sn\ N\right\}.
\]

A nerve complex endowed with a proximal relator is a collection of spokes with proximities given in Lemma~\ref{lem:MNCnerves} and Theorem~\ref{thm:sn-nerve}.

\begin{lemma}\label{lem:MNCnerves}
Let $\Nrv K$ be a nerve complex endowed with the strong proximity $\sn$.  Then $\mathop{\bigcap}\limits_{\sk A\in \Nrv K} \sk A\neq \emptyset$.
\end{lemma}
\begin{proof}
Let $\sk A,\sk B$ be a pair of spokes in nerve $\Nrv K$.  Since $\sk A,\sk B$ have a nucleus $N$ in common, $\sk A\ \sn \sk B$ implies $\sk A\cap \sk B\neq \emptyset$ (from Axiom (snN2)).  This holds true for all spokes in $\Nrv K$.  Consequently, $\mathop{\bigcap}\limits_{\sk A\in \Nrv K} \sk A\neq \emptyset$.
\end{proof}

\begin{theorem}\label{thm:sn-nerve}
Let $\left(\Nrv K,\left\{\near,\dnear,\sn\right\}\right)$ be a proximal relator space, spokes $\sk A,\sk B\in \Nrv K$.  Then 
\begin{compactenum}[1$^o$]
\item $\sk A\ \sn\ \sk B \Rightarrow \sk A\ \near\ \sk A$.
\item $\sk A\ \sn\ \sk B \Rightarrow \sk A\ \dnear\ \sk B$.
\end{compactenum}
\end{theorem}
\begin{proof}$\mbox{}$\\
1$^o$: From Lemma~\ref{lem:MNCnerves}, $\mbox{Nrv}\mathscr{F}_{_{MNC}}$ is an Edelsbrunner-Harer nerve.  

$\sk A\ \sn\ \sk B$ for every pair of spokes $\sk A,\sk B$ in the nerve $\Nrv K$, since $\sk A,\sk B$ have common vertex, namely, the nucleus of $\Nrv K$.  From Lemma~\ref{lem:MNCnerves}, $\mathop{\bigcap}\limits_{\sk A\in \Nrv K} \sk A\neq \emptyset$ for every pair of spokes in $\Nrv K$.  Consequently, $\sk A\ \sn\ \sk B$.  Then, from Axiom (snN2), $\sk A\ \near\ \sk A$ for all spokes $\sk A,\sk B$.\\
2$^o$: Spokes $\sk A,\sk B$ have nucleus $N$ in common.  Hence, $\sk A\ \dcap\ \sk B\neq \emptyset$.  Then, from Lemma~\ref{thm:sn-implies-near}, $\sk A\ \dcap\ \sk B \neq \emptyset \Rightarrow\ \sk A\ \dnear\ \sk B$. This gives the desired result for each pair of spokes in the nerve.
\end{proof}

\begin{theorem}\label{EHnerve}{\rm ~\cite[\S III.2, p. 59]{Edelsbrunner1999}}
Let $\mathscr{F}$ be a finite collection of closed, convex sets in Euclidean space.  Then the nerve of $\mathscr{F}$ and the union of the sets in $\mathscr{F}$ have the same homotopy type.
\end{theorem}

\begin{lemma}\label{thm:1spokeHomotopy}
Let $\Nrv K$ be a finite collection of 1-spokes, which are closed, convex sets in a triangulation of a region in the Euclidean plane.  The nerve $\Nrv K$ and the union of the 1-spokes in $\Nrv K$ have the same homotopy type.
\end{lemma}
\begin{proof}
From Theorem~\ref{EHnerve}, we have that the union of the 1-spokes $\sk A\in \Nrv K$ and $\Nrv K$ have the same homotopy type.
\end{proof}

\begin{remark}
Every 1-spoke in a nerve complex $\Nrv K$ is part of an $n$-spoke, $n > 1$.  From a consideration of $n$-spoke extensions of 1-spokes in $\Nrv K$, we obtain the main result of this paper, namely,
Theorem~\ref{thm:nerveSpokeTheorem} as a straightforward corollary of Lemma~\ref{thm:1spokeHomotopy}.
\qquad \textcolor{blue}{\Squaresteel}
\end{remark}
   
\begin{theorem}\label{lem:sndMNCnerves}
Let $2^X$ be a finite collection of strongly near maximal nerve complexes covering a finite region of the Euclidean plane equipped with the relator $\left\{\sn,\snd\right\}$, $\Phi(\Nrv K)$ = number of 1-spokes in $\Nrv K$.   Then $\mathop{\bigcap}\limits_{\Phi}\Nrv K \neq \emptyset$.
\end{theorem}
\begin{proof}
Each maximal $\Nrv K$ has the same description, namely, the number of 1-spokes in the nerve complex.  Hence, $\mathop{\bigcap}\limits_{\Phi}\Nrv K \neq \emptyset$.
\end{proof}

\section{Application: Detecting Image Object Shapes}
This section carries forward the notion of descriptively proximal images.  The study image object shapes is aided by detecting triangulation nerves containing the maximal number of filled triangles (denoted by $\maxNerv K$).

\begin{remark} {\bf How to Detect Image Object Shapes with Maximal Nerve Clusters}.\\
The following steps lead to the detection of image object shapes.
\begin{description}
\item[Triangulation]: The triangulation of a digital image depends on the initial choice of vertices, used as generating points for filled triangles.  In this work, keypoints have been chosen.
A keypoint is defined by the gradient orientation (an angle) and gradient magnitudes (edge strength) of each image pixel.  All selected image keypoints have different gradient orientations and edge strengths.  Typically, in an image with a 100,000 pixels, we might find 1000 keypoints.  Keypoints are ideally suited for shape detection, since each keypoint is also usually an edge pixel.  
\item[Nerve Complexes]: Keypoint-based nerve complexes have a nucleus that is a keypoint. 
For nerve complexes, see 
Fig.~\ref{fig:girlNerves}.
\item[Maximal Nerves]: In the search for a principal image object shape, we start with $\maxNerv K$ in a nerve with the maximum number of filled triangles that have a nucleus vertex in common.  Experimentally, it has been found a large part of an image object shape will be covered by a maximal nerve (see, {\em e.g.},~\cite[\S 8.9,\S 9.1]{Peters2017ComputerVision}).
\item[Maximal Nerve Clusters]: Nerve complexes strongly near $\maxNerv K$ form a cluster useful in detecting an image object shape (denoted by $\maxNrvClu K$).\\
\item[Shape Contour]: The outer perimeter of a $\maxNrvClu K$ provides the contour of a shape that can be compared with other known shape contours, leading to the formation of shape classes.  A $\maxNrvClu K$ outer perimeter is called a shape edgelet~\cite[\S 7.6]{Peters2017ComputerVision}.  Shape contour comparisons can be accomplished by decomposing contour into nerve complexes and extracting geometric features of the nerves, {\em e.g.}, nerve centroids, area (sum of the areas of the nerve triangles), number of nerve triangles, maximal nerve triangle area, which are easily compared across triangulations of different images.
\end{description}
\qquad \textcolor{blue}{\Squaresteel}
\end{remark}

 \begin{figure}[!ht]
	\centering
	\subfigure[Vietri girl]{
	 \includegraphics[width=35mm]{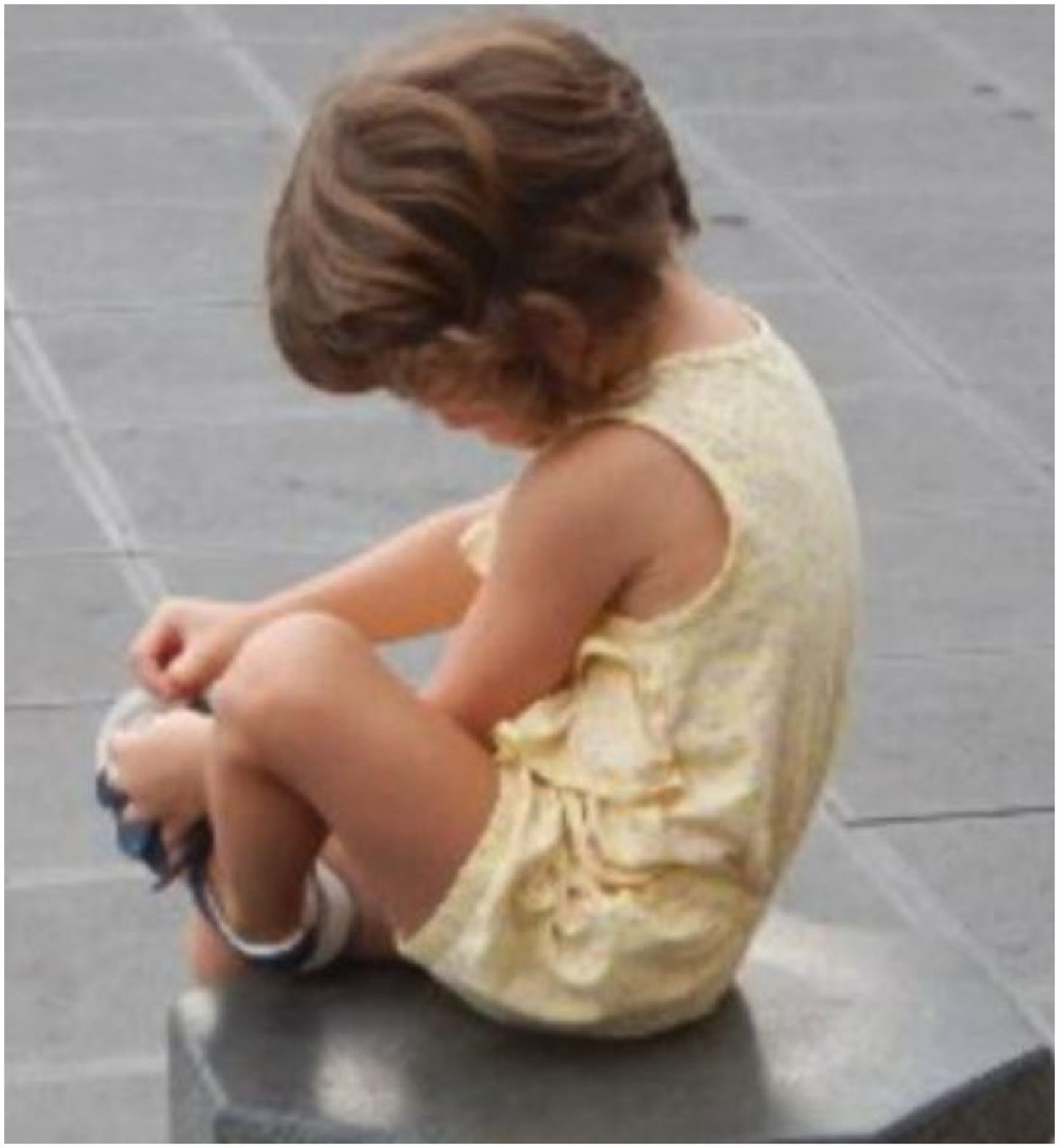}\label{fig:girl}}\hfil
	\subfigure[Girl vertices]{
	 \includegraphics[width=35mm]{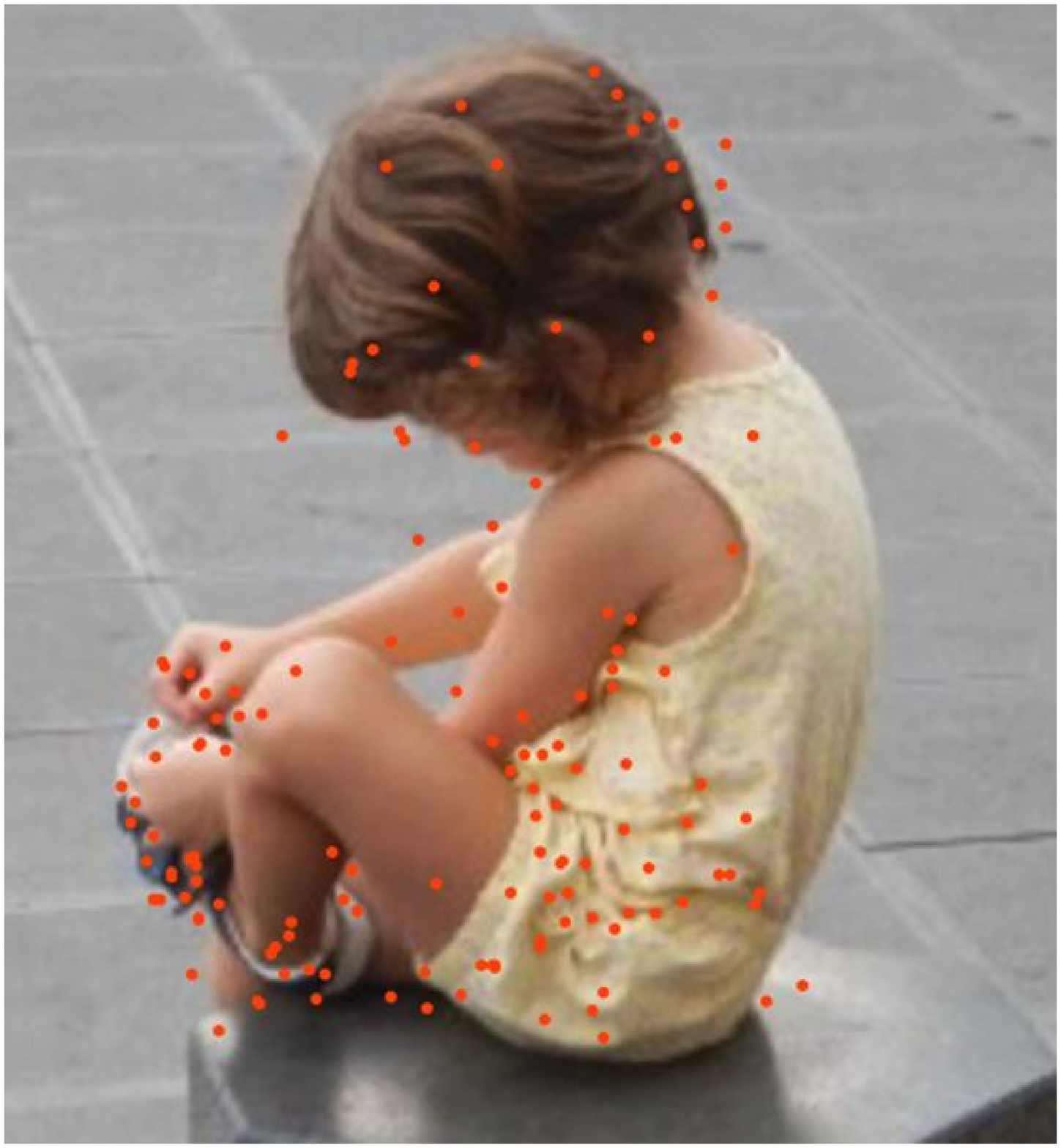}\label{fig:girlVertices}}
	\subfigure[Girl nerves]{
	 \includegraphics[width=35mm]{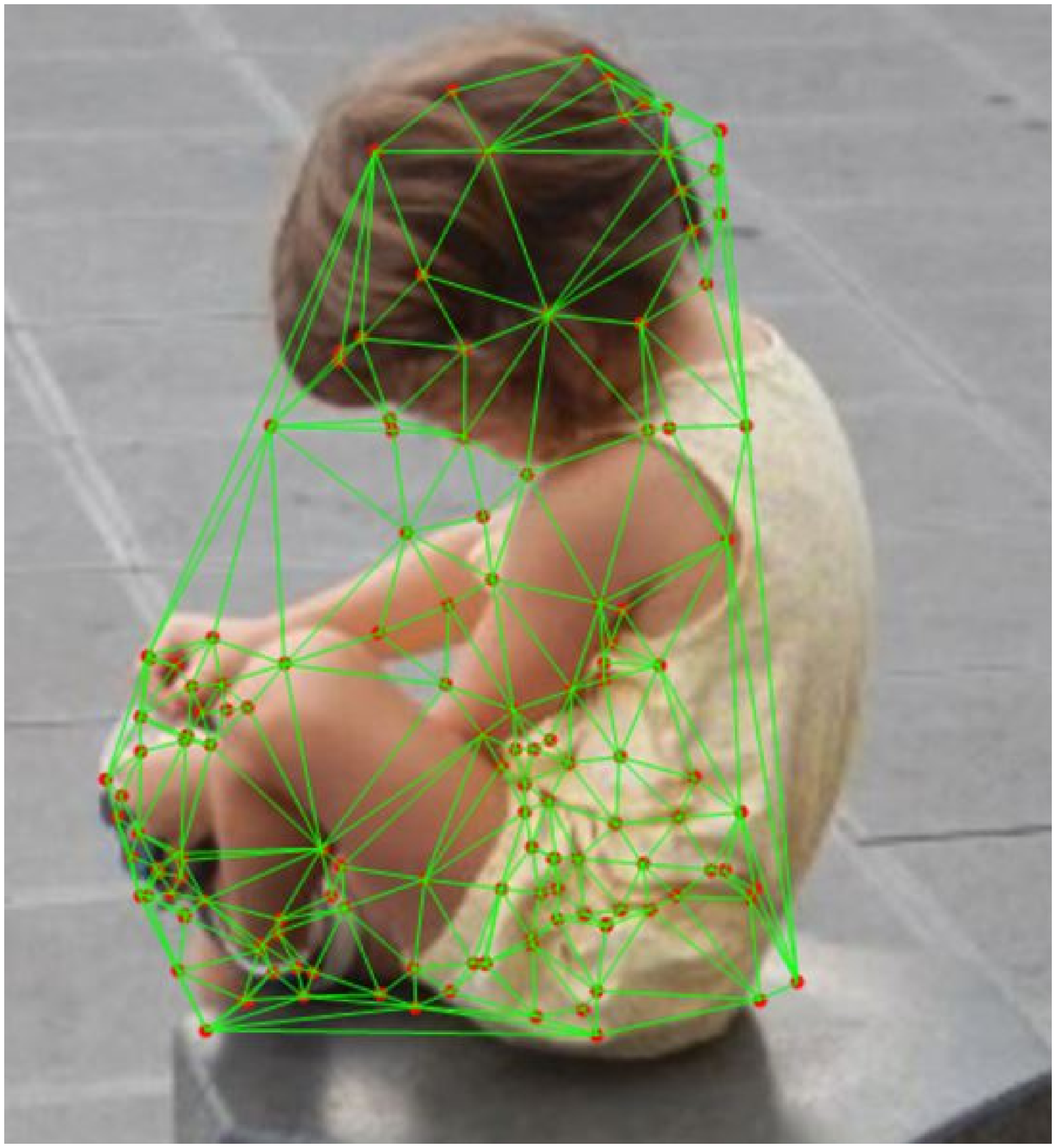}\label{fig:girlNerves}}
	\caption{Simplex vertices and overlapping Girl nerve simplicial complexes}
	\label{fig:girlGeometry}
\end{figure}

\begin{figure}[ht]
\centering
\includegraphics[width=90mm]{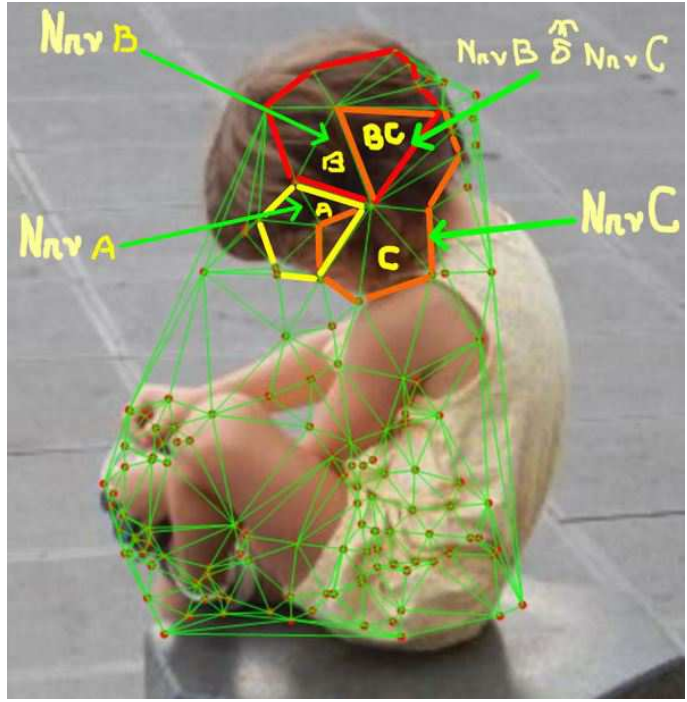}
\caption[]{Overlapping Nerve Complexes on Girl Image}
\label{fig:girlNrvs}
\end{figure}

\begin{example} {\bf Girl Closure Nerve Cluster Head Shape}.\\
Based on a selection of keypoints shown in Fig.~\ref{fig:girlVertices}, the triangulation of the girl image in Fig.~\ref{fig:girl} leads to a collection of filled triangle simplexes that cover the central region of the image as shown in Fig.~\ref{fig:girlNerves}.  From this triangulation, maximal nucleus clusters can be identified.  For example, we can begin to detect the shape of the head from the collection of overlapping nerve complexes in Fig.~\ref{fig:girlNrvs}.
The nerve complexes in Fig.~\ref{fig:girlNrvs} form a cluster with $\maxNerv A$ doing most of the work in highlighting the shape of a large part of the girl's head.  Let the upper region of this space be endowed with what is known as proximal relator~\cite{Peters2016FilomatProximalRelator}, which is a collection of proximity relations on the space.  Let $\left(X,\mathscr{R}_{\delta}\right)$ be a proximal relator space with the upper region of Fig.~\ref{fig:girl} represented by the set of points $X$ and let $\mathscr{R}_{\delta} = \left\{\sn,\dnear,\snd\right\}$.  Notice that the triangulation in Fig.~\ref{fig:girlNerves} contains a number of separated triangles.  Hence, from Theorem~\ref{lem:nerveFamily}, we can expect to find more than one nerve.  In this case, families of nerves can be found in this image, starting with the upper region of the image.  Then observe the following things in the collections of nerve complexes in Fig.~\ref{fig:girlNrvs}.
\begin{align*}
\Nrv A\ & \sn\ \Nrv B\ \mbox{\rm (Nerves with a common 2-spoke are strongly near (Theorem~\ref{lem:stronglyNearNerves}))},\\
\Nrv A\ & \sn\ \Nrv C\ \mbox{\rm (From Theorem~\ref{lem:stronglyNearNerves}, these nerves are strongly near)},\\
\Nrv B\ & \sn\ \Nrv C\ \mbox{\rm (From Theorem~\ref{lem:stronglyNearNerves}, these nerves are strongly near)},\\
\left(\Nrv A\ \cap\ \Nrv B\right)\ & \dnear\ \Nrv C\ \mbox{\rm (nerves with matching feature vectors, cf. Theorem~\ref{lem:descriptivelyNearNerves})},\\
\left(Nrv A\ \cap\ \Nrv C\right)\ & \dnear\ \Nrv B\ \mbox{\rm (nerves with matching feature vectors, cf. Theorem~\ref{lem:descriptivelyNearNerves})},\\
\left(\Nrv A\ \cap\ \Nrv C\right)\ & \snd\ \Nrv B\ \mbox{\rm (nerve interiors with matching descriptions, cf. Theorem~\ref{lem:stronglyDescriptivelyNearNerves})},\\
\left(\Nrv A\ \cap\ \Nrv B\right)\ & \snd\ \Nrv C\ \mbox{\rm (nerve interiors with matching descriptions, cf. Theorem~\ref{lem:stronglyDescriptivelyNearNerves})}.
\end{align*}
From these proximities, we can derive the head shape from the contour formed by the sequence of connected line segments along the outer edges of the nerve spokes.
\qquad \textcolor{blue}{\Squaresteel}
\end{example}

\section*{Concluding Remarks}
A theory of proximal nerve complexes is introduced in this paper.  An application of this theory is in the form of a framework for the detection of image object shapes.  For other promising places for applications of proximal nerve complexes, see, {\em e.g.},
\cite{Tozzi2017NsciLettersBarcodes},~\cite{Tozzi2017LambertMultidimensionalWorld} and~\cite[\S 5.3, \S 5.4, \S 5.10, \S 12.1 and \S 14.1]{Peters2016CP}. 
  
\bibliographystyle{amsplain}
\bibliography{NSrefs}

\end{document}